%% file: 2col.tex
\documentclass[journal,twocolumn,10pt]{IEEEtran}
\usepackage{graphicx}
\usepackage{epstopdf}
\usepackage{amsmath}
\usepackage{amssymb}
\usepackage{amsfonts}
\usepackage{mathrsfs}
\usepackage{amsthm}
\usepackage{bbm}
\usepackage{multirow}
\usepackage{subfig}
\usepackage{chgbar}
\newtheorem{lem}{Lemma}[section]

\begin{document}
\title{Physical Layer Security of Generalised Pre-coded Spatial Modulation with Antenna Scrambling}

\author{Rong Zhang, Lie-Liang Yang and Lajos Hanzo \\
Communications, Signal Processing and Control, School of ECS, University of Southampton, SO17 1BJ, UK \\
Email: {rz,lly,lh}@ecs.soton.ac.uk, http://www-mobile.ecs.soton.ac.uk\protect
\thanks{{The financial support of the EPSRC under the India-UK Advanced Technology Centre (IU-ATC), that of the EU under the Concerto project as well as that of the European Research Council's (ERC) Advanced Fellow Grant is gratefully acknowledged.}}
}

\markboth{}{Physical Layer Security of Generalised Pre-coded Spatial Modulation with Antenna Pattern Confusion}

\maketitle

%

\input{content}
\bibliographystyle{IEEEtran}
\bibliography{ref,rong_pub}

\end{document}

%% file: content.tex
\section{Introduction}
In the last decade or so, physical layer security had been an active research area drawing substantial attentions and collaborative efforts across different societies. Complementary to the conventional application layer security approach dominantly based on cryptography that is potentially vulnerable to quantum computationally capable eavesdropper, physical layer security was envisioned of providing the ultimate security from the information theoretical point of view. This area of research was pioneered by Shannon's foundation work on secrete communications~\cite{6769090}, where the 'one-time pad' was proposed in order to achieve the perfect security. This landmark finding inspired further intriguing research on the physical layer security of 'wire-tap' channel~\cite{6772207}, where the eavesdropper's channel between Alice and Eve constitutes a degraded version of the main channel between Alice and Bob. In particular, the important notion of security capacity was formally defined as the capacity difference between the main channel and the eavesdropper's channel, where a positive security capacity enables physical layer security. 

The research of physical layer security became much more complicated yet interesting when wireless communications is considered, where the eavesdropper's channel becomes not necessarily degraded~\cite{1055892,4529282,4626059} in comparison to the main channel owing to the fast variations of the wireless channel propagation. As a benefit, the classic solution exploited the fast varying Channel State Information (CSI) of the main channel as the security key~\cite{1255548}, which was assumed to be only \textit{privately} known between Alice and Bob. This solution became naturally inappropriate in the graver condition, where the CSI of the main channel became also \textit{publicly} known to Eve. In this challenging scenario, various signal processing advances appeared particularly in the so-called Multiple Input Multiple Output Multiple Eavesdropper (MIMOME) physical layer security model~\cite{5485016,5605343}. More explicitly, when eavesdropper's channel is optimistically assumed to be known to Alice having multiple antennas, intelligent beam-forming constituted effective solutions. On the other hand, when eavesdropper's channel is unknown to Alice but its dimension (number of antennas) is known, the technique of generating artificial noise~\cite{4543070} by Alice was found to be powerful, where the artificial noise was pre-processed to lie in the null space of the main channel so as to only affect Eve adversely~\footnote{Similarly, exploiting helper to generate artificial interference towards Eve in the cooperative communications scenario is also found effective~\cite{4608977,5673779}.}. However, one common limitation of the above two distinct approaches was the stringent \textit{dimensionality} requirement, where Alice must have more antennas than Eve.

In contrast to the above two design philosophies, we propose a physical layer security solution, without the necessity of obeying the dimensionality requirement, based on our previously proposed Generalised Pre-coded Spatial Modulation (GPSM) scheme~\cite{5956573,6644231}. Let us now briefly introduce the background of GPSM scheme before embarking on our novel design. Spatial Modulation (SM) constituted novel Multiple Input Multiple Output (MIMO) technique, which conveys extra information by appropriately activating the \textit{Transmit} Antenna (TA) indices, in addition to the classic modulation schemes used on each activated TAs, as detailed in~\cite{4382913,6678765}. By contrast, the GPSM scheme is capable of conveying extra information by appropriately selecting the \textit{Receive} Antenna (RA) indices in addition to the classic modulation schemes loaded on each activated RAs, as detailed in~\cite{6644231}. To elaborate a little more, in~\cite{6644231}, comprehensive performance comparisons were carried out between the GPSM scheme as well as the conventional MIMO scheme and the associated detection complexity issues were discussed. Furthermore, a range of practical issues were investigated, namely the detrimental effects of realistic imperfect Channel State Information at the Transmitter (CSIT), followed by a low-rank approximation invoked for large-dimensional MIMOs. As a further exploitation, error probability, achievable rate and capacity analysis were conducted in~\cite{Zhang}. When physical layer security is considered, early work on SM considered its straightforward use, where the security capacity is firstly reported under various SNRs~\cite{6399113} and then compared to that of the single input single output fading channel~\cite{5956508}. 

We now advocate a novel physical layer security solution that is unique to our previously proposed GPSM scheme with the aid of the proposed \textit{antenna scrambling}. The novelty and contribution of our paper lies in three aspects:
\begin{itemize}
\item \textit{principle}: we introduce a 'security key' generated at Alice that is \textit{unknown} to both Bob and Eve, where the design goal is that the publicly unknown security key only imposes barrier for Eve.
\item \textit{approach}: we achieve it by conveying useful information only through the activation of RA indices, which is in turn concealed by the unknown security key in terms of the randomly scrambled symbols used in place of the conventional modulated symbols in GPSM scheme. 
\item \textit{design}: we consider both Circular Antenna Scrambling (CAS) and Gaussian Antenna Scrambling (GAS) in detail and the resultant security capacity of both designs are quantified and compared. 
\end{itemize}
As a design merit, we find that the proposed physical layer security solution imposes great challenge on Eve and provides positive security capacity under both the conventional setting of having less antennas at Eve than at Alice as well as the disadvantageous setting that Alice is less equipped than Eve, which is largely neglected in the literature.

The organisation of the paper is as follows. In Section~\ref{sec_systemmodel}, we introduce the physical layer security model considered, along with the brief introduction of GPSM for self completeness in addition to the straightforward use of GPSM scheme in terms of physical layer security. This motivates our proposed antenna scrambling in Section~\ref{sec_antennascrambling}, where both CAS and GAS are discussed in detail together with their security capacity expressions. We present numerical results in Section~\ref{sec_nr} and conclude in Section~\ref{sec_con}.

\section{Physical Layer Security of GPSM}
\label{sec_systemmodel}
Consider a typical physical layer security scenario with $N_t$ TAs at Alice and $N_r$ RAs at Bob in the presence of $N_e$ RAs at Eve, where we assume $N_t \geq N_r$ with no constraints put on $N_e$. Throughout the paper, we further assume that the MIMO channel between Alice and Bob $\pmb{H}_{B} \in \mathbb{C}^{N_r\times N_t}$ is known between each other and is also publicly available at Eve, while the MIMO channel between Alice and Eve $\pmb{H}_{E} \in \mathbb{C}^{N_e\times N_t}$ is only known at Eve, where both $\pmb{H}_{B}$ and $\pmb{H}_{E}$ are frequency-flat Rayleigh fading with each entry being independent. In this MIMOME set-up, we have the following typical system of equations constituted by the coexistence of legal and illegal communications
\begin{align}
\pmb{y}_B &= \pmb{H}_B\pmb{x}+\pmb{w}_B, \label{eq_bob}\\
\pmb{y}_E &= \pmb{H}_E\pmb{x}+\pmb{w}_E, \label{eq_eve}
\end{align}
where $\pmb{w}_B \in \mathbb{C}^{N_r\times 1}$ and $\pmb{w}_E \in \mathbb{C}^{N_e\times 1}$ are the circularly symmetric complex Gaussian noise vector with each entry having a zero mean and a variance of $\sigma^2_B$ and $\sigma^2_E$, i.e. we have $\mathbb{E}[||\pmb{w}_B||^2] = N_r\sigma^2_B$ and $\mathbb{E}[||\pmb{w}_E||^2] = N_e\sigma^2_E$. The noise $\pmb{w}_B$ and $\pmb{w}_E$ are independent between each other and are also independent from the transmit signal $\pmb{x} \in \mathbb{C}^{N_t\times 1}$. Finally, the security capacity $C_S$ of the above MIMOME set-up is defined as the capacity difference between Alice and Bob $C_B$ and Alice and Eve $C_E$, which is formerly expressed as
\begin{align}
C_S &= C_B - C_E.
\label{eq_sc}
\end{align}

\subsection{Transceiver Structure of GPSM}
\label{sec_tagpsm}
Let us now elaborate more on our previously proposed GPSM scheme.

\subsubsection{Conceptual Description}
In our GPSM scheme, a total of $N_a < N_r$ RAs are activated at Bob so as to facilitate the simultaneous transmission of $N_a$ data streams between Alice and Bob, where the particular \textit{pattern} of the $N_a$ RAs activated conveys extra information in form of so-called spatial symbols in addition to the information carried by the conventional modulated symbols mapped to them. Hence, the number of bits in GPSM conveyed by a spatial symbol becomes $k_{ant} = \lfloor\log_2(|\mathcal{C}_t|)\rfloor$, where the set $\mathcal{C}_t$ contains all the combinations associated with choosing $N_a$ activated RAs out of $N_r$ RAs. As a result, the total number of bits transmitted by the GPSM scheme is $k_{eff} = k_{ant}+N_ak_{mod}$, where $k_{mod} = \log_2(M)$ denotes the number of bits per symbol of a conventional $M$-ary modulation scheme and its alphabet is denoted by $\mathcal{A}$. Finally, it is plausible that the conventional MIMO scheme between Alice and Bob with multiplexing of $N_a = N_r$ data streams conveys a total of $k_{eff} = N_rk_{mod}$ bits altogether. For assisting further discussions, we also let $\mathcal{C} \subset \mathcal{C}_t$ be the set of selected activation patterns, $\mathcal{C}(k)$ and $\mathcal{C}(k,i)$ denote the $k$th RA activation pattern and the $i$th activated RA in the $k$th activation pattern, respectively. 

\subsubsection{Mathematical Description}
Let $\pmb{s}^k_m$ be an \textit{explicit} representation of a so-called super-symbol $\pmb{s} \in \mathbb{C}^{N_r\times 1}$, indicating that the RA pattern $k$ at Bob is activated and $N_a$ conventional modulated symbols $\pmb{b}_m = [b_{m_1},\ldots,b_{m_{N_a}}]^T \in \mathbb{C}^{N_a\times 1}$ are transmitted, where we have $b_{m_i} \in \mathcal{A}$ and $\mathbb{E}[|b_{m_i}|^2] = 1, \forall i \in [1,N_a]$. In other words, we have the relationship 
\begin{align}
\pmb{s} &= \pmb{s}^k_m = \pmb{\Omega}_k \pmb{b}_m,
\label{eq_txs}
\end{align}
where $\pmb{\Omega}_k = \pmb{I}[:,\mathcal{C}(k)]$ is constituted by the specifically selected columns determined by $\mathcal{C}(k)$ of an identity matrix of $\pmb{I}_{N_r}$. Following pre-coding of $\pmb{P} \in \mathbb{C}^{N_t\times N_r}$, the resultant transmit signal $\pmb{x}$ from Alice may be written as 
\begin{align}
\pmb{x} &= \sqrt{\beta / N_a} \pmb{P} \pmb{s},
\label{eq_txx}
\end{align}
where $\beta$ is designed for maintaining $\mathbb{E}[||\pmb{x}||^2] = 1$. The conventional design of GPSM allows no energy leaks into the unintended RA patterns at Bob. Hence, the classic linear Channel Inversion (CI)-based pre-coding~\cite{1391204,1261332} may be used, which is formulated as
\begin{align}
\pmb{P} &= \pmb{H}_B^H (\pmb{H}_B \pmb{H}_B^H)^{-1},
\label{eq_ci}
\end{align}
where the power-normalisation factor of the output power after pre-coding is $\beta = N_r/\mathrm{Tr}[(\pmb{H}_B \pmb{H}_B^H)^{-1}]$. 

As a result, the signal observed at the $N_r$ RAs of Bob may be written as 
\begin{align}
\pmb{y}_B &= \sqrt{\beta / N_a} \pmb{H}_B \pmb{P} \pmb{s}^k_m+\pmb{w}_B.
\label{eq_rx}
\end{align}
Consequently, the joint detection of both the conventional modulated symbols $\pmb{b}_m$ and of the spatial symbol $k$ obeys the Maximum Likelihood (ML) criterion, which is formulated as
\begin{align}
[\hat{m}_1, \ldots, \hat{m}_{N_a},\hat{k}] &= \arg \min_{\pmb{s}^\ell_n \in \mathcal{S}} ||\pmb{y}_B - \sqrt{\frac{\beta}{N_a}} \pmb{H}_B \pmb{P}\pmb{s}^\ell_n||^2 ,
\label{eq_joint}
\end{align}
where $\mathcal{S} = \mathcal{C} \times \mathcal{A}^{N_a}$ denotes the super-alphabet. Alternatively, decoupled or separate detection may also be employed, which treats the detection of the conventional modulated symbols $\pmb{b}_m$ and the spatial symbol $k$ separately. In this reduced-complexity variant, we have
\begin{align}
\hat{k} &= \arg \max_{\ell \in [1,|\mathcal{C}|]} \sum^{N_a}_{i=1}|y^B_{\mathcal{C}(\ell,i)}|^2 , \label{eq_seperate1} \\
\hat{m}_i &= \arg \min_{n_i \in [1,M]} |y^B_{\mathcal{C}(\hat{k},i)} - \sqrt{\frac{\beta}{N_a}} \pmb{h}^B_{\mathcal{C}(\hat{k},i)} \pmb{p}_{\mathcal{C}(\hat{k},i)} b_{n_i}|^2,
\label{eq_seperate2}
\end{align}
where $\pmb{h}^B_{\mathcal{C}(\hat{k},i)}$ is the $\mathcal{C}(\hat{k},i)$th row of $\pmb{H}_B$ representing the channel between the $\mathcal{C}(\hat{k},i)$th RA at Bob and the transmitter at Alice, while $\pmb{p}_{\mathcal{C}(\hat{k},i)}$ is the $\mathcal{C}(\hat{k},i)$th column of $\pmb{P}$. Thus, correct detection is declared, when we have $\hat{k} = k$ and $\hat{m}_i = m_i, \forall i$~\footnote{Although our security capacity discussion later does not rely on a particular detection scheme, we include them for self-completeness.}.

\subsection{Security Capacity of GPSM}
\label{sec_scgpsm}
Since the MIMO channel $\pmb{H}_E$ is intentionally and practically assumed unknown to Alice, we now investigate the security capacity of GPSM scheme dispensing with the knowledge of $\pmb{H}_E$ at Alice.

\subsubsection{DCMC Capacity}
The security capacity as expressed in (\ref{eq_sc}) for GPSM scheme relies on the Discrete input Continuous output Memoryless Channel (DCMC) capacity calculations of $C_B$ and $C_E$, respectively. It is known that Shannon's channel capacity is obtained, when the input signal obeys a Gaussian distribution~\cite{ThomasM.Cover2006}. Our GPSM scheme is special in the sense that the spatial symbol conveys integer values constituted by the RA pattern index, which does not obey the shaping requirements of Gaussian signalling. This implies that the channel capacity of the GPSM scheme depends on a mixture of a continuous and a discrete input. Hence, for simplicity's sake, we discuss the DCMC capacity in the context of discrete-input signalling for both the spatial symbol and for the conventional modulated symbols mapped to it. 

Before calculating $C_B$ and $C_E$ for (\ref{eq_bob}) and (\ref{eq_eve}) employing GPSM scheme, we first provide a general expression of calculating the DCMC capacity that will be used throughout the paper. The DCMC capacity is obtained by maximizing the mutual information between discrete input $\pmb{u} \in \mathcal{B}$ and continuous output $\pmb{z}$, where we have
\begin{align}
\hspace*{-0.25cm}
C &= \max_{p(\pmb{u}_{\tau}), \forall \tau} I(\pmb{u};\pmb{z}) \nonumber \\
&= \max_{p(\pmb{u}_{\tau}), \forall \tau} \sum^{|\mathcal{B}|}_{\tau = 1} \int\limits^{\infty}_{-\infty} p(\pmb{z},\pmb{u}_{\tau}) \log_2 \left( \frac{p(\pmb{z}|\pmb{u}_{\tau})}{\sum^{|\mathcal{B}|}_{\epsilon = 1} p(\pmb{z},\pmb{u}_{\epsilon})}\right) d \pmb{z},
\label{eq_dcmc_t}
\end{align}
furthermore, we have
\begin{align}
&\log_2 \left( \frac{p(\pmb{z}|\pmb{u}_{\tau})}{\sum^{|\mathcal{B}|}_{\epsilon = 1} p(\pmb{z},\pmb{u}_{\epsilon})}\right) = \log_2 \left( \frac{p(\pmb{z}|\pmb{u}_{\tau})}{\sum^{|\mathcal{B}|}_{\epsilon = 1} p(\pmb{z}|\pmb{u}_{\epsilon})p(\pmb{u}_{\epsilon})}\right) \nonumber \\ 
&= -\log_2 \left( \frac{1}{|\mathcal{B}|} \sum^{|\mathcal{B}|}_{\epsilon = 1} \frac{p(\pmb{z}|\pmb{u}_{\epsilon})}{p(\pmb{z}|\pmb{u}_{\tau})} \right) = \log_2(|\mathcal{B}|)-\log_2 \sum^{|\mathcal{B}|}_{\epsilon = 1} \Theta_{\epsilon,\tau},
\label{eq_term}
\end{align}
where we let $\Theta_{\epsilon,\tau} = p(\pmb{z}|\pmb{u}_{\epsilon}) / p(\pmb{z}|\pmb{u}_{\tau})$. By substituting (\ref{eq_term}) into (\ref{eq_dcmc_t}) and exploiting that (\ref{eq_dcmc_t}) is maximized when we have equal probable $p(\pmb{u}_{\tau})$, we have
\begin{align}
C &= \log_2(|\mathcal{B}|)-\frac{1}{|\mathcal{B}|} \sum^{|\mathcal{B}|}_{\tau = 1} \mathbb{E} \left[ \log_2 \sum^{|\mathcal{B}|}_{\epsilon = 1} \Theta_{\epsilon,\tau} \right].
\label{eq_dcmc_tf}
\end{align}

\subsubsection{Security Capacity}
\label{sec_scpgpsm}
By substituting (\ref{eq_txx}) into (\ref{eq_bob}) and (\ref{eq_eve}), we have the explicit relation
\begin{align}
\pmb{y}_B &= \sqrt{\beta / N_a} \pmb{H}_B \pmb{P} \pmb{s} +\pmb{w}_B, \label{eq_bobp}\\
\pmb{y}_E &= \sqrt{\beta / N_a} \pmb{H}_E \pmb{P} \pmb{s} +\pmb{w}_E. \label{eq_evep}
\end{align}
Hence, with perfect knowledge of $\pmb{G}_B = \sqrt{\beta / N_a} \pmb{H}_B \pmb{P}$ and $\pmb{G}_E = \sqrt{\beta / N_a} \pmb{H}_E \pmb{P}$ being the equivalent channels, the conditional probability of receiving $\pmb{y}_B$ and $\pmb{y}_E$ given that $\pmb{s}_{\tau}$ was transmitted are formulated as
\begin{align}
p(\pmb{y}_B|\pmb{s}_{\tau}) &\propto \exp\left\lbrace \frac{-||\pmb{y}_B-\pmb{G}_B\pmb{s}_{\tau}||^2}{\sigma_B^2}\right\rbrace, \label{eq_cond_bp} \\
p(\pmb{y}_E|\pmb{s}_{\tau}) &\propto \exp\left\lbrace \frac{-||\pmb{y}_E-\pmb{G}_E\pmb{s}_{\tau}||^2}{\sigma_E^2}\right\rbrace. \label{eq_cond_ep}
\end{align}
As a result, we have 
\begin{align}
\Theta^B_{\epsilon,\tau} = \exp\left\lbrace \frac{-||\pmb{G}_B(\pmb{s}_{\tau}-\pmb{s}_{\epsilon})+\pmb{w}_B||^2+||\pmb{w}_B||^2}{\sigma_B^2} \right\rbrace, \label{eq_theta_bp} \\
\Theta^E_{\epsilon,\tau} = \exp\left\lbrace \frac{-||\pmb{G}_E(\pmb{s}_{\tau}-\pmb{s}_{\epsilon})+\pmb{w}_E||^2+||\pmb{w}_E||^2}{\sigma_E^2} \right\rbrace, \label{eq_theta_ep}
\end{align}
and hence the DCMC capacity of $C_B$ and $C_E$ for (\ref{eq_bobp}) and (\ref{eq_evep}) may be obtained upon substituting (\ref{eq_theta_bp}) and (\ref{eq_theta_ep}) into (\ref{eq_dcmc_tf}), respectively and also by letting $|\mathcal{B}| = |\mathcal{C}|M^{N_a}$ in (\ref{eq_dcmc_tf}). Finally, the security capacity is obtained by evaluating (\ref{eq_sc}) as a difference between $C_B$ and $C_E$.

\subsubsection{Performance}
Fig~\ref{fig_GPSM} shows the security capacity of GPSM scheme employing CI based pre-coding when $[N_t,N_r] = [8,4]$ and using QPSK modulation having different number of activated RAs $N_A = [1,2,3]$ when $N_e = N_r = 4$ (dash) and $N_e = N_t = 8$ (solid). It is clearly shown in Fig~\ref{fig_GPSM} that when $N_e = N_r = 4$, the security capacity is positive for all $N_a$ considered and activating more RAs results into higher security capacity. On the other hand, no positive security capacity is found when $N_e = N_t = 8$ for all $N_a$ considered as shown by solid curves. Indeed, $N_e = N_t$ constitutes a challenge setting in the physical layer security literature, where typically $N_e < N_t$ was considered. This is because even though the eavesdropper's channel is known by Alice mysteriously, the celebrated strategies of transmitting in the null space of Eve or of minimizing the signal leakage towards Eve leave no degree of freedom for communicating between Alice and Bob if $N_e < N_t$ was violated. This is also true for the typical strategy of generating artificial noise by Alice, where the underlining assumption of $N_e < N_t$ was also explicitly imposed. 

\begin{figure}[t]
\centering
\includegraphics[width=\linewidth]{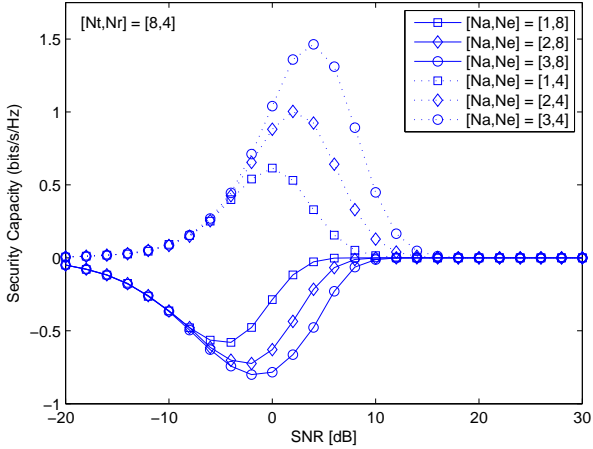}
\caption{Security capacity of GPSM scheme employing CI based pre-coding when $[N_t,N_r] = [8,4]$ and using QPSK modulation having different number of activated RAs $N_A = [1,2,3]$ when $N_e = N_r = 4$ (dash) and $N_e = N_t = 8$ (solid).}
\label{fig_GPSM}
\end{figure}

\section{Security Via Antenna Scrambling}
\label{sec_antennascrambling}
We now propose a novel physical layer security scheme \textit{exclusive} to GPSM scheme with the aid of antenna scrambling. 

\subsection{Motivation and Principle}
In the security capacity of GPSM analysed in Section~\ref{sec_scgpsm} and showed in Fig~\ref{fig_GPSM}, we found that the explicit and complete knowledge of $\mathcal{B}$ known to both Bob and Eve constitutes one potential source of confidentiality leakage, which allows Eve to perform exhaustive search over $\mathcal{B}$ when Eve has unbounded computational capability. 

Based on this backdrop, we propose to \textit{only} convey useful information through the RA patterns in terms of spatial symbols, while the conventional modulated symbols $\pmb{b}_m$ belonging to a finite alphabet of $\mathcal{A}^{N_a}$ are purposely replaced by randomly generated scrambled symbols, which are changed at symbol rate and are unknown to both Bob and Eve. In other words, the useful information in spatial domain is concealed by the time-domain scrambled symbols with infinite possibilities, which are in turn acted as security key.

Thus, the aim is to investigate proper scrambled symbols that result in significant barriers only for Eve to retrieve the useful information in spatial domain, but not for Bob. In this light, we consider both Circular Antenna Scrambling (CAS) and Gaussian Antenna Scrambling (GAS), where the scrambled symbols in the former and in the latter are drawn from a unit circle with uniformly distributed phases and from a complex Gaussian distribution with zero-mean and (normalised) unity variance, respectively. 

\subsection{Security Capacity of GPSM with CAS}
When CAS is considered, we rewrite (\ref{eq_txs}) with conventional modulated symbols $\pmb{b}_m$ replaced by random symbols $\pmb{e} = [e^{j\theta_1},\ldots,e^{j\theta_{N_a}}]^T \in \mathbb{C}^{N_a\times 1}$ drawn from a unit circle with uniformly distributed phase $\theta_i, \forall i \in [1,N_a]$, namely we have
\begin{align}
\pmb{s} &= \pmb{\Omega}_{\tau} \pmb{e}, 
\label{eq_txcas}
\end{align}
when the $\tau$th RA pattern is activated at Bob. It is thus clear that the useful information is contained in $\pmb{\Omega}_{\tau}$ only and is concealed by the phase rotations of $\pmb{e}$. 

Now, if we followed the same procedure as in Section~\ref{sec_scpgpsm} and adopted the same equations as expressed in (\ref{eq_bobp}) and (\ref{eq_evep}) by treating $\pmb{s}$ as a whole in order to obtain the DCMC capacity of $C_B$ and $C_E$, the conditional probability calculations of (\ref{eq_cond_bp}) and (\ref{eq_cond_ep}) become problematic. This is because the symbol $\pmb{s}$ of (\ref{eq_txcas}) belongs to an infinite alphabet in CAS, where its enumeration is impossible. 

On the other hand, we may only consider the spatial domain information contained in $\pmb{\Omega}_{\tau}$ with finite alphabet of $\mathcal{C}$. However, the unknown randomly scrambled symbols of $\pmb{e}$, which is also changed at symbol rate, analogously contribute to strong channel uncertainty to both Bob and Eve. This makes the conditional probability calculations of (\ref{eq_cond_bp}) and (\ref{eq_cond_ep}) become again intractable since coherent detection is not allowed with unknown $\pmb{e}$. 

We thus discuss the security capacity of GPSM scheme with CAS replying on \textit{non-coherent} DCMC capacity calculations of $C_B$ and $C_E$ based on the non-coherent transformation of (\ref{eq_bobp}) and (\ref{eq_evep}) with (\ref{eq_txcas}) been substituted, which are given as follows
\begin{align}
\pmb{r}_B &= |\pmb{G}_B \pmb{\Omega}_{\tau} \pmb{e} +\pmb{w}_B|^2 = |\sideset{}{_{v=1}^{N_a}}\sum \pmb{g}^B_{\mathcal{C}(\tau,v)} e^{j\theta_v} +\pmb{w}_B|^2,  \label{eq_bob_non} \\
\pmb{r}_E &= |\pmb{G}_E \pmb{\Omega}_{\tau} \pmb{e} +\pmb{w}_E|^2 = |\sideset{}{_{v=1}^{N_a}}\sum \pmb{g}^E_{\mathcal{C}(\tau,v)} e^{j\theta_v} +\pmb{w}_E|^2,  \label{eq_eve_non}
\end{align}
where we have $\pmb{g}^B_{\mathcal{C}(\tau,v)}$ and $\pmb{g}^E_{\mathcal{C}(\tau,v)}$ stand for the $\mathcal{C}(\tau,v)$th column of $\pmb{G}_B$ and $\pmb{G}_E$ respectively. Let us further elaborate a bit more on two cases as follows.

\subsubsection{$N_a = 1$}
In this special case, we have a simplified expression of (\ref{eq_bob_non}) and (\ref{eq_eve_non}) given as
\begin{align}
\pmb{r}_B &= |\pmb{g}^B_{\mathcal{C}(\tau,1)} e^{j\theta_1} +\pmb{w}_B|^2,  \label{eq_bob_non_s} \\
\pmb{r}_E &= |\pmb{g}^E_{\mathcal{C}(\tau,1)} e^{j\theta_1} +\pmb{w}_E|^2,  \label{eq_eve_non_s}
\end{align}
Hence, by normalising with respect to $\sigma^2_{B,0} = \sigma^2_{B}/2$ and $\sigma^2_{E,0} = \sigma^2_{E}/2$, it is clear that for the $i$th entry of $\pmb{r}_B$ and $\pmb{r}_E$, we have 
\begin{align}
r_{B,i}/\sigma^2_{B,0} &\sim \chi^2_2 \left[\lambda_{B,i} = |g^B_{\mathcal{C}(\tau,1),i}|^2/\sigma^2_{B,0} \right], \\
r_{E,i}/\sigma^2_{E,0} &\sim \chi^2_2 \left[\lambda_{E,i} = |g^E_{\mathcal{C}(\tau,1),i}|^2/\sigma^2_{E,0} \right], 
\end{align}
where $\chi^2_2(\lambda_{B,i})$ and $\chi^2_2(\lambda_{E,i})$ stand for the non-central chi-square distribution with degree of freedom of two and non-centrality given by $|g^B_{\mathcal{C}(\tau,1),i}|^2/\sigma^2_{B,0}$ and $|g^E_{\mathcal{C}(\tau,1),i}|^2/\sigma^2_{E,0}$ as a result of $|e^{j\theta_1}|^2 = 1$, respectively. In particular, when CI pre-coding of (\ref{eq_ci}) is employed, where 
\begin{align}
\pmb{G}_B = \sqrt{\beta / N_a} \pmb{H}_B \pmb{P} = \sqrt{\beta / N_a} \pmb{I}_{N_r}, 
\label{eq_bs}
\end{align}
we have $\lambda_{B,i} = 0, \forall i \not\in \mathcal{C}(\tau)$ for inactivated RAs and $\lambda_{B,i} = \beta/N_a\sigma^2_{B,0}, \forall i \in \mathcal{C}(\tau)$ for activated RAs. 

Thus the conditional probability of receiving $\pmb{r}_B$ and $\pmb{r}_E$ given that $\pmb{\Omega}_{\tau}$ was transmitted and subjected to the knowledge of equivalent channels $\pmb{G}_B$ and $\pmb{G}_E$ but dispensing with the knowledge of $\pmb{e}$ are formulated as
\begin{align}
p(\pmb{r}_B|\pmb{\Omega}_{\tau}) &= \prod_{i=1}^{N_r} f_{\chi^2_2 \left[|g^B_{\mathcal{C}(\tau,1),i}|^2/\sigma^2_{B,0} \right]} \left( \frac{r_{B,i}}{\sigma^2_{B,0}} \right), \label{eq_cond_non_b} \\
p(\pmb{r}_E|\pmb{\Omega}_{\tau}) &= \prod_{i=1}^{N_e} f_{\chi^2_2 \left[|g^E_{\mathcal{C}(\tau,1),i}|^2/\sigma^2_{E,0} \right]} \left( \frac{r_{E,i}}{\sigma^2_{E,0}} \right). \label{eq_cond_non_e}
\end{align}
As a result, we have 
\begin{align}
\hspace*{-0.25cm}
\Theta^B_{\epsilon,\tau} &= \frac{\sideset{}{_{i=1}^{N_r}} \prod f_{\chi^2_2 \left[|g^B_{\mathcal{C}(\epsilon,1),i}|^2/\sigma^2_{B,0} \right]} \left( r_{B,i}/\sigma^2_{B,0} \right)}{\sideset{}{_{i=1}^{N_r}} \prod f_{\chi^2_2 \left[|g^B_{\mathcal{C}(\tau,1),i}|^2/\sigma^2_{B,0} \right]} \left( r_{B,i}/\sigma^2_{B,0} \right)}, \label{eq_theta_non_b} \\
\Theta^E_{\epsilon,\tau} &= \frac{\sideset{}{_{i=1}^{N_e}} \prod f_{\chi^2_2 \left[|g^E_{\mathcal{C}(\epsilon,1),i}|^2/\sigma^2_{E,0} \right]} \left( r_{E,i}/\sigma^2_{E,0} \right)}{\sideset{}{_{i=1}^{N_r}} \prod f_{\chi^2_2 \left[|g^E_{\mathcal{C}(\tau,1),i}|^2/\sigma^2_{E,0} \right]} \left( r_{E,i}/\sigma^2_{E,0} \right)}, \label{eq_theta_non_e}
\end{align}
and hence when $N_a = 1$, the non-coherent DCMC capacity of $C_B$ and $C_E$ for (\ref{eq_bob_non_s}) and (\ref{eq_eve_non_s}) may be obtained by letting $|\mathcal{B}| = |\mathcal{C}|$ and substituting (\ref{eq_theta_non_b}) and (\ref{eq_theta_non_e}) into (\ref{eq_dcmc_tf}). Finally, the security capacity is obtained by evaluating (\ref{eq_sc}).

\subsubsection{$N_a > 1$}
Now we consider (\ref{eq_bob_non}) and (\ref{eq_eve_non}) when $N_a > 1$, where in this case we have 
\begin{align}
\hspace*{-0.25cm}
r_{B,i}/\sigma^2_{B,0} &\sim \chi^2_2\left[\lambda_{B,i} = | \sideset{}{_{v=1}^{N_a}}\sum g^B_{\mathcal{C}(\tau,v),i} e^{j\theta_v} |^2/\sigma^2_{B,0} \right], \label{eq_dist_bob_cas} \\
\hspace*{-0.25cm}
r_{E,i}/\sigma^2_{E,0} &\sim \chi^2_2\left[\lambda_{E,i} = | \sideset{}{_{v=1}^{N_a}}\sum g^E_{\mathcal{C}(\tau,v),i} e^{j\theta_v} |^2/\sigma^2_{E,0} \right]. \label{eq_dist_eve_cas}
\end{align}
In contrast to $N_a = 1$, since $\pmb{e}$ is unknown to both Bob and Eve, the non-centrality of both $\lambda_{B,i}$ and $\lambda_{E,i}$ of (\ref{eq_dist_bob_cas}) and (\ref{eq_dist_eve_cas}) become also unavailable in general. This is because the impact of the unknown $\pmb{e}$ adds extra constructive and destructive fading. As a result, the conditional probability calculations similar to (\ref{eq_cond_non_b}) and (\ref{eq_cond_non_e}) may be intractable. 

In particular, when the CI pre-coding of (\ref{eq_ci}) is employed with the beneficial structure of (\ref{eq_bs}), we have $\lambda_{B,i} = 0, \forall i \not\in \mathcal{C}(\tau)$ for inactivated RAs and $\lambda_{B,i} = \beta/N_a\sigma^2_{B,0}, \forall i \in \mathcal{C}(\tau)$ for activated RAs. Hence, the knowledge of $\pmb{e}$ becomes \textit{unnecessary} for Bob in order to evaluate (\ref{eq_dist_bob_cas}). However, the beneficial structure similar to (\ref{eq_bs}) is not seen for Eve as $\pmb{H}_E\pmb{P} \in \mathcal{C}^{N_e \times N_r}$, so the knowledge of $\pmb{e}$ becomes \textit{crucial} for Eve in order to evaluate (\ref{eq_dist_eve_cas}), resulting into a barrier. 

Consequently, with the knowledge of both $\pmb{H}_B$ and $\pmb{H}_E$, Eve is forced to be 'pretended' as Bob by conducting post-processing on $\pmb{y}_E$ of (\ref{eq_evep}) as
\begin{align}
\pmb{\tilde{y}}_{E} &=\pmb{H}_B(\pmb{H}^H_E\pmb{H}_E)^{-1}\pmb{H}^H_E \pmb{y}_E, \nonumber \\
&=\pmb{G}_B\pmb{\Omega}_{\tau} \pmb{e}+\pmb{H}_B(\pmb{H}^H_E\pmb{H}_E)^{-1}\pmb{H}^H_E\pmb{w}_E,
\label{eq_eve_post_cas}
\end{align}
where by letting $\pmb{\tilde{w}}_E = \pmb{H}_B(\pmb{H}^H_E\pmb{H}_E)^{-1}\pmb{H}^H_E\pmb{w}_E$, the non-coherent transformation of (\ref{eq_eve_post_cas}) is 
\begin{align}
\pmb{\tilde{r}}_{E} &=|\pmb{G}_B\pmb{\Omega}_{\tau} \pmb{e}+\pmb{\tilde{w}}_E|^2 = |\sideset{}{_{v=1}^{N_a}}\sum \pmb{g}^B_{\mathcal{C}(\tau,v)} e^{j\theta_v} +\pmb{\tilde{w}}_E|^2.
\label{eq_eve_non_post_cas}
\end{align}
Hence, by being forced to equip higher number of antennas than Alice ($N_e \geq N_t$) in order to perform the above post-processing of (\ref{eq_eve_post_cas}), Eve finally becomes also capable of retrieving the information concealed in the spatial domain dispensing with the knowledge of $\pmb{e}$ at the cost of noise amplification provided $\sigma^2_E > 0$. More explicitly, by normalising with respect to $\tilde{\sigma}^2_{E,0} = \tilde{\sigma}^2_{E}/2$, we have $\tilde{r}_{E,i}/\tilde{\sigma}^2_{E,0} \sim \chi^2_2(\tilde{\lambda}_{E,i})$ with $\tilde{\lambda}_{E,i} = 0, \forall i \not\in \mathcal{C}(\tau)$ for inactivated RAs and $\tilde{\lambda}_{E,i} = \beta/N_a\tilde{\sigma}^2_{E,0}, \forall i \in \mathcal{C}(\tau)$ for activated RAs.  

Thus, when the CI pre-coding of (\ref{eq_ci}) is employed and $N_e \geq N_t$, the conditional probability of receiving $\pmb{r}_B$ and $\pmb{\tilde{r}}_{E}$ given that $\pmb{\Omega}_{\tau}$ was transmitted and subjected to the knowledge of equivalent channels $\pmb{G}_B$ and $\pmb{G}_E$ but dispensing with the knowledge of $\pmb{e}$ are formulated as
\begin{align}
\hspace*{-0.25cm}
p(\pmb{r}_B|\pmb{\Omega}_{\tau}) &= \prod_{i=1}^{N_r} f_{\chi^2_2 \left[ \mathbbm{1}[i \in \mathcal{C}(\tau)]\beta/N_a\sigma^2_{B,0} \right]} \left( \frac{r_{B,i}}{\sigma^2_{B,0}} \right), \label{eq_cond_non_b1} \\
\hspace*{-0.25cm}
p(\pmb{\tilde{r}}_E|\pmb{\Omega}_{\tau}) &= \prod_{i=1}^{N_r} f_{\chi^2_2\left[ \mathbbm{1}[i \in \mathcal{C}(\tau)]\beta/N_a\tilde{\sigma}^2_{E,0}\right]} \left( \frac{\tilde{r}_{E,i}}{\tilde{\sigma}^2_{E,0}} \right), \label{eq_cond_non_e1}
\end{align}
where $\mathbbm{1}[\cdot]$ is the indicator function that returns one (zero) when the entry is true (false). As a result, we have 
\begin{align}
\Theta^B_{\epsilon,\tau} &= \frac{\sideset{}{_{i=1}^{N_r}}\prod f_{\chi^2_2 \left[ \mathbbm{1}[i \in \mathcal{C}(\epsilon)]\beta/N_a\sigma^2_{B,0} \right]} \left( r_{B,i}/\sigma^2_{B,0} \right)}{\sideset{}{_{i=1}^{N_r}}\prod f_{\chi^2_2 \left[ \mathbbm{1}[i \in \mathcal{C}(\tau)]\beta/N_a\sigma^2_{B,0} \right]} \left( r_{B,i}/\sigma^2_{B,0} \right)}, \label{eq_theta_non_b2} \\
\Theta^E_{\epsilon,\tau} &= \frac{\sideset{}{_{i=1}^{N_r}}\prod f_{\chi^2_2\left[ \mathbbm{1}[i \in \mathcal{C}(\epsilon)]\beta/N_a\tilde{\sigma}^2_{E,0}\right]} \left( \tilde{r}_{E,i}/\tilde{\sigma}^2_{E,0} \right)}{\sideset{}{_{i=1}^{N_r}}\prod f_{\chi^2_2\left[ \mathbbm{1}[i \in \mathcal{C}(\tau)]\beta/N_a\tilde{\sigma}^2_{E,0}\right]} \left( \tilde{r}_{E,i}/\tilde{\sigma}^2_{E,0} \right)}, \label{eq_theta_non_e2}
\end{align}
and hence when $N_a > 1$ and CI pre-coding is employed with $N_e \geq N_t$, the non-coherent DCMC capacity of $C_B$ and $C_E$ for (\ref{eq_bob_non}) and (\ref{eq_eve_non_post_cas}) may be obtained by letting $|\mathcal{B}| = |\mathcal{C}|$ and substituting (\ref{eq_theta_non_b2}) and (\ref{eq_theta_non_e2}) into (\ref{eq_dcmc_tf}). Then, the security capacity is obtained by evaluating (\ref{eq_sc}). Finally, note that when $N_e < N_t$, Eve becomes not capable of gaining any information contained in the RA patterns without knowing $\pmb{e}$ even though $\sigma^2_E \mapsto 0$. 

\subsection{Security Capacity of GPSM with GAS}
When GAS is considered, we rewrite (\ref{eq_txs}) with conventional modulated symbols $\pmb{b}_m$ replaced by random symbols $\pmb{n} = [n_1,\ldots,n_{N_a}]^T \in \mathbb{C}^{N_a\times 1}$ drawn from a complex Gaussian distribution with zero-mean and (normalised) unity variance, i.e. $n_i \in \mathcal{CN}(0,1)$, namely we have
\begin{align}
\pmb{s} &= \pmb{\Omega}_{\tau} \pmb{n}, 
\label{eq_txgas}
\end{align}
when the $\tau$th RA pattern is activated at Bob. Now, when compared with the CAS case, the uncertainties arise from both phase and amplitude. From certain perspective, we may treat GAS as introducing multiplicative artificial noise. 

Similar to CAS, since the symbol $\pmb{s}$ belongs to an infinite alphabet in GAS, we may consider the finite alphabet of $\mathcal{C}$ for the spatial domain information contained in $\pmb{\Omega}_{\tau}$. But unlike in CAS, now coherent DCMC capacity calculations of $C_B$ and $C_E$ are possible. By substituting (\ref{eq_txgas}) into (\ref{eq_bobp}) and (\ref{eq_evep}), we have
\begin{align}
\pmb{y}_B &= \pmb{G}_B \pmb{\Omega}_{\tau} \pmb{n} +\pmb{w}_B = \sideset{}{_{v=1}^{N_a}}\sum \pmb{g}^B_{\mathcal{C}(\tau,v)} n_v +\pmb{w}_B,  \label{eq_bob_g} \\
\pmb{y}_E &= \pmb{G}_E \pmb{\Omega}_{\tau} \pmb{n} +\pmb{w}_E = \sideset{}{_{v=1}^{N_a}}\sum \pmb{g}^E_{\mathcal{C}(\tau,v)} n_v +\pmb{w}_E.  \label{eq_eve_g}
\end{align}
Then it is clear that for the $i$th entry of $\pmb{y}_B$ and $\pmb{y}_E$ as expressed in (\ref{eq_bob_g}) and (\ref{eq_eve_g}), we have 
\begin{align}
y_{B,i} \sim \mathcal{CN}(0,\sideset{}{_{v=1}^{N_a}}\sum |g^B_{\mathcal{C}(\tau,v),i}|^2 + \sigma^2_B),  \\
y_{E,i} \sim \mathcal{CN}(0,\sideset{}{_{v=1}^{N_a}}\sum |g^E_{\mathcal{C}(\tau,v),i}|^2 + \sigma^2_E), 
\end{align}
where without the explicit knowledge of $\pmb{n}$, the conditional probability of receiving $\pmb{y}_B$ and $\pmb{y}_{E}$ given that $\pmb{\Omega}_{\tau}$ was transmitted and subjected to the knowledge of equivalent channels $\pmb{G}_B$ and $\pmb{G}_E$ are formulated as
\begin{align}
p(\pmb{y}_B|\pmb{\Omega}_{\tau}) &\propto \exp \left\lbrace -\sum_{i=1}^{N_r} \frac{|y_{B,i}|^2}{\sum_{v=1}^{N_a} |g^B_{\mathcal{C}(\tau,v),i}|^2+\sigma^2_B}\right\rbrace, \\
p(\pmb{y}_E|\pmb{\Omega}_{\tau}) &\propto \exp \left\lbrace -\sum_{i=1}^{N_e} \frac{|y_{E,i}|^2}{\sum_{v=1}^{N_a} |g^E_{\mathcal{C}(\tau,v),i}|^2 + \sigma^2_E}\right\rbrace.
\end{align}
As a result, we have 
\begin{align}
\hspace*{-0.25cm}
\Theta^B_{\epsilon,\tau} &= \frac{\exp \left\lbrace \displaystyle \sum_{i=1}^{N_r} |y_{B,i}|^2 / \left(\sum_{v=1}^{N_a} |g^B_{\mathcal{C}(\tau,v),i}|^2+\sigma^2_B \right) \right\rbrace}{\exp \left\lbrace \displaystyle \sum_{i=1}^{N_r} |y_{B,i}|^2 / \left(\sum_{v=1}^{N_a} |g^B_{\mathcal{C}(\epsilon,v),i}|^2+\sigma^2_B \right) \right\rbrace}, \label{eq_theta_gas_b}\\
\hspace*{-0.25cm}
\Theta^E_{\epsilon,\tau} &= \frac{\exp \left\lbrace \displaystyle \sum_{i=1}^{N_e} |y_{E,i}|^2 / \left(\sum_{v=1}^{N_a} |g^E_{\mathcal{C}(\tau,v),i}|^2+\sigma^2_E \right) \right\rbrace}{\exp \left\lbrace \displaystyle \sum_{i=1}^{N_e} |y_{E,i}|^2 / \left(\sum_{v=1}^{N_a} |g^E_{\mathcal{C}(\epsilon,v),i}|^2+\sigma^2_E \right) \right\rbrace},\label{eq_theta_gas_e}
\end{align}
and hence dispensing with the knowledge of $\pmb{n}$, the DCMC capacity of $C_B$ and $C_E$ for (\ref{eq_bob_g}) and (\ref{eq_eve_g}) may be obtained by letting $|\mathcal{B}| = |\mathcal{C}|$ and substituting (\ref{eq_theta_gas_b}) and (\ref{eq_theta_gas_e}) into (\ref{eq_dcmc_tf}). Finally, the security capacity is obtained by evaluating (\ref{eq_sc}).

However, a further look results into the following Lemma
\begin{lem} 
The DCMC capacity calculation of $C_B$ based on (\ref{eq_theta_gas_b}) for Bob is $0 \leq C_B \leq \log_2(|\mathcal{C}|)$, but the DCMC capacity calculation of $C_E$ based on (\ref{eq_theta_gas_e}) returns zero capacity for Eve, i.e. we have $C_E = 0$.
\label{lemma}
\end{lem}
\begin{proof}
Let us rewrite (\ref{eq_dcmc_tf}) here by setting $|\mathcal{B}| = |\mathcal{C}|$ as
\begin{align}
C &= \log_2(|\mathcal{C}|)-\frac{1}{|\mathcal{C}|} \sum^{|\mathcal{C}|}_{\tau = 1} \mathbb{E} \left[ \log_2 \sum^{|\mathcal{C}|}_{\epsilon = 1} \Theta_{\epsilon,\tau} \right],
\label{eq_rewrite}
\end{align}
and we further expand the following term as
\begin{align}
&\mathbb{E} \left[ \log_2 \sum^{|\mathcal{C}|}_{\epsilon = 1} \Theta_{\epsilon,\tau} \right] = \mathbb{E} \left[ - \log_2 \frac{1}{\sum^{|\mathcal{C}|}_{\epsilon = 1} \Theta_{\epsilon,\tau}} \right]  \nonumber \\
&\geq -\log_2  \mathbb{E} \left[ \frac{1}{\sum^{|\mathcal{C}|}_{\epsilon = 1} \Theta_{\epsilon,\tau}} \right] \geq -\log_2 \frac{1}{\sum^{|\mathcal{C}|}_{\epsilon = 1}  \mathbb{E} \left[\Theta_{\epsilon,\tau}  \right]} \nonumber \\
&= \log_2 \sum^{|\mathcal{C}|}_{\epsilon = 1}  \mathbb{E} \left[ \Theta_{\epsilon,\tau}  \right] = \log_2 \sum^{|\mathcal{C}|}_{\epsilon = 1}  \mathbb{E} \left[ e^{A-B} \right] \nonumber \\
&\geq \log_2 \left( 1+\sum^{|\mathcal{C}|}_{\epsilon = 1, \epsilon \neq \tau}  e^{\mathbb{E} \left[ A-B \right]} \right),
\label{eq_expand_important} 
\end{align}
where the inequities involved during the above derivations are followed from Jensen's inequality. Further, A and B are shorthand for the inner term in the exponential operator of both the numerator and denominator, respectively, either of (\ref{eq_theta_gas_b}) or of (\ref{eq_theta_gas_e}). Substituting (\ref{eq_expand_important}) into (\ref{eq_rewrite}), we have
\begin{align}
C &\leq \log_2(|\mathcal{C}|)-\frac{1}{|\mathcal{C}|} \sum^{|\mathcal{C}|}_{\tau = 1} \log_2 \left( 1+\sum^{|\mathcal{C}|}_{\epsilon = 1, \epsilon \neq \tau}  e^{\mathbb{E} \left[ A-B \right]} \right)
\label{eq_expand}
\end{align}

1) When $C_B$ is considered, let us first define $\mathcal{C}_c = \mathcal{C}(\tau) \cap \mathcal{C}(\epsilon)$ as the set containing the commonly activated RAs of two different activation patterns and thus $\mathcal{C}_d = \bar{\mathcal{C}}_c$ denotes the complementary set containing the activated RAs in difference. We further define $\mathcal{C}^u(\tau) = \mathcal{C}(\tau) \cap \mathcal{C}_d$ and $\mathcal{C}^u(\epsilon) = \mathcal{C}(\epsilon) \cap \mathcal{C}_d$ as the set containing the activated RAs uniquely found in $\mathcal{C}(\tau)$ and $\mathcal{C}(\epsilon)$, respectively. With the above definitions and by exploiting the beneficial property of (\ref{eq_bs}), we have
\begin{align}
\mathbb{E} \left[ A-B \right] &= \sum^{|\mathcal{C}_d|/2}_{i=1} \frac{\mathbb{E} \left[\Upsilon_i\right]}{\beta/N_a+\sigma^2_B}-\frac{\mathbb{E} \left[\Upsilon_i\right]}{\sigma^2_B}. 
\label{eq_expand_theta}
\end{align}
where we have $\Upsilon_i = |y^B_{\mathcal{C}^u(\tau,i)}|^2-|y^B_{\mathcal{C}^u(\epsilon,i)}|^2$ and $\mathbb{E} \left[\Upsilon_i\right] = \beta/N_a$. Now, when $\sigma^2_B \mapsto 0$, the summation of (\ref{eq_expand_theta}) approaches $-\infty$, so (\ref{eq_expand_important}) approaches zero and (\ref{eq_expand}) results into $C_B \leq \log_2(|\mathcal{C}|)$. On the other hand, when $\sigma^2_B \mapsto \infty$, the summation of (\ref{eq_expand_theta}) approaches zero, so (\ref{eq_expand_important}) approaches $\log_2 (|\mathcal{C}|)$ and (\ref{eq_expand}) results into $C_B \leq 0$. Since we also must have $C_B \geq 0$, we end up with $C_B = 0$. Overall, we have $0 \leq C_B \leq \log_2(|\mathcal{C}|)$.

2) To see $C_E = 0$, it is simple to see from (\ref{eq_theta_gas_e}) that
\begin{align}
&\mathbb{E} \left[ A-B \right] = \sum^{N_e}_{i=1} \mathbb{E} \left[ \frac{|y_{E,i}|^2}{A'+\sigma^2_E} \right]  - \mathbb{E} \left[ \frac{|y_{E,i}|^2}{B'+\sigma^2_E} \right] = 0,
\label{eq_expand_theta_e}
\end{align}
where $A' = \sum_{v=1}^{N_a} |g^E_{\mathcal{C}(\tau,v),i}|^2$ and $B' = \sum_{v=1}^{N_a} |g^E_{\mathcal{C}(\epsilon,v),i}|^2$ and we have $\mathbb{E} \left[ A'\right] = \mathbb{E} \left[ B'\right]$. As a result, (\ref{eq_expand_important}) approaches $\log_2 (|\mathcal{C}|)$ and hence (\ref{eq_expand}) results into $C_E \leq 0$. Since we also must have $C_E \geq 0$, we end up with $C_E = 0$.
\end{proof} 
\noindent Lemma~\ref{lemma} suggests that Eve cannot follow the same approach to retrieve the useful information as Bob. Thus, similarly to the context of CAS, by being forced to equip more antennas than Alice ($N_e \geq N_t$) and at the cost of noise amplification provided $\sigma^2_E > 0$, Eve is again tried to be 'pretended' as Bob by conducting post-processing on $\pmb{y}_E$ of (\ref{eq_eve_g}) to arrive at
\begin{align}
\pmb{\tilde{y}}_{E} &=\pmb{G}_B\pmb{\Omega}_{\tau} \pmb{n}+\pmb{H}_B(\pmb{H}^H_E\pmb{H}_E)^{-1}\pmb{H}^H_E\pmb{w}_E.
\label{eq_eve_g2}
\end{align}  
Hence, we have
\begin{align}
\tilde{y}_{E,i} &\sim \mathcal{CN}(0,\sideset{}{_{v=1}^{N_r}}\sum |g^B_{\mathcal{C}(\tau,v),i}|^2 + \tilde{\sigma}^2_E), 
\end{align}
and after a few manipulations, we arrive at
\begin{align}
\Theta^E_{\epsilon,\tau} &= \frac{\exp \left\lbrace \displaystyle \sum_{i=1}^{N_r} |\tilde{y}_{E,i}|^2 / \left(\sum_{v=1}^{N_a} |g^B_{\mathcal{C}(\tau,v),i}|^2+\tilde{\sigma}^2_E \right) \right\rbrace}{\exp \left\lbrace \displaystyle \sum_{i=1}^{N_r} |\tilde{y}_{E,i}|^2 / \left(\sum_{v=1}^{N_a} |g^B_{\mathcal{C}(\epsilon,v),i}|^2+\tilde{\sigma}^2_E \right) \right\rbrace}, \label{eq_theta_gas_e2}
\end{align}
and hence dispensing with the knowledge of $\pmb{n}$, the DCMC capacity of $C_E$ for (\ref{eq_eve_g2}) may be obtained by letting $|\mathcal{B}| = |\mathcal{C}|$ and substituting (\ref{eq_theta_gas_e2}) into (\ref{eq_dcmc_tf}). Then, the security capacity is obtained by evaluating (\ref{eq_sc}). Finally, again note that when $N_e < N_t$, Eve becomes not capable of gaining any information contained in the RA patterns without knowing $\pmb{e}$ even though $\sigma^2_E \mapsto 0$, since the operation of (\ref{eq_eve_g2}) is not allowed. 

\subsection{Further Discussions}
\begin{table}[t]
\centering
\begin{tabular}{ |c|c|c|c|c| }
\hline
 & GPSM & \multicolumn{2}{ |c| }{CAS} & GAS \\ \hline
\multirow{2}{*}{$\Theta^B_{\epsilon,\tau}$, $\Theta^E_{\epsilon,\tau}$} & \multirow{2}{*}{(\ref{eq_theta_bp}),(\ref{eq_theta_ep})} & $N_a = 1$ & $N_a > 1$ & \multirow{2}{*}{(\ref{eq_theta_gas_b}), (\ref{eq_theta_gas_e2})} \\ \cline{3-4}
& & (\ref{eq_theta_non_b}), (\ref{eq_theta_non_e}) & (\ref{eq_theta_non_b2}), (\ref{eq_theta_non_e2}) & \\ \hline
\end{tabular}
\caption{Summarized findings of $\{ \Theta^B_{\epsilon,\tau}, \Theta^E_{\epsilon,\tau} \}$.}
\label{table}
\end{table}

\subsubsection{Example}
We now provide an example to better aid the conceptual understanding of our proposed physical layer security of GPSM with antenna scrambling. Fig~\ref{fig_casgas} shows the plot of 1000 samples of the received signal $\pmb{y}_B$ at Bob and $\pmb{y}_E$ at Eve when CI pre-coding aided GPSM is employed at Alice with $[N_t,N_r,N_e] = [8,4,4]$ under a particular MIMO channel realisation of $\pmb{H}_B$ and $\pmb{H}_E$ when $\mathrm{SNR}$ is set to 30dB to remove the disturbance of noise. Thus we focus on the effect of CAS and GAS on the realisations of $\pmb{y}_B$ and $\pmb{y}_E$. In Fig~\ref{fig_casgas}, we include three different cases correspond to three rows of subplots in pair for $\pmb{y}_B$ (left) and $\pmb{y}_E$ (right), respectively. On each subplot, the received signal samples of each individual RA are detailed in four miniplots, each with horizontal axis and vertical axis represent real and imaginary part of $y_{B,i}, \forall i \in [1,N_r]$ or $y_{E,i}, \forall i \in [1,N_e]$. 

The first row of Fig~\ref{fig_casgas} illustrate the scenario, where $N_a = 1$ and employing CAS with first RA been activated. It can be seen from Fig~\ref{fig_casb1} that the received signal samples of the first RA at Bob is clearly distinguishable with respect to that of the rest RAs. Moreover, being a circle at the first RA implies that its non-coherent transformation (magnitude) constitutes a constant value. On the other hand, the four distinguishable circles observed at Eve from Fig~\ref{fig_case1} imply that Eve is also unaffected by CAS, but failing to capture the correct pattern as compared to Bob. 

When the second row of Fig~\ref{fig_casgas} is considered, where $N_a > 1$ and employing CAS with first and second RAs been activated. It can be seen from Fig~\ref{fig_casb} that the received signal samples of the first two RAs are stand-out clearly from the rest, being unaffected by the unknown knowledge of $\pmb{e}$. However, Eve is completely messed up by the confusion of CAS as seen from Fig~\ref{fig_case}, where a large overlap can be found in these four miniplots. Finally, we consider the third row of Fig~\ref{fig_casgas} when GAS is employed. Again, clear patterns are observed at Bob as seen in Fig~\ref{fig_gasb} but indistinguishable clouds are found at Eve as seen in Fig~\ref{fig_gase}.

  \begin{figure}[t]
    \subfloat[$N_a = 1$ CAS (Bob) \label{fig_casb1}]{%
      \includegraphics[width=0.485\linewidth]{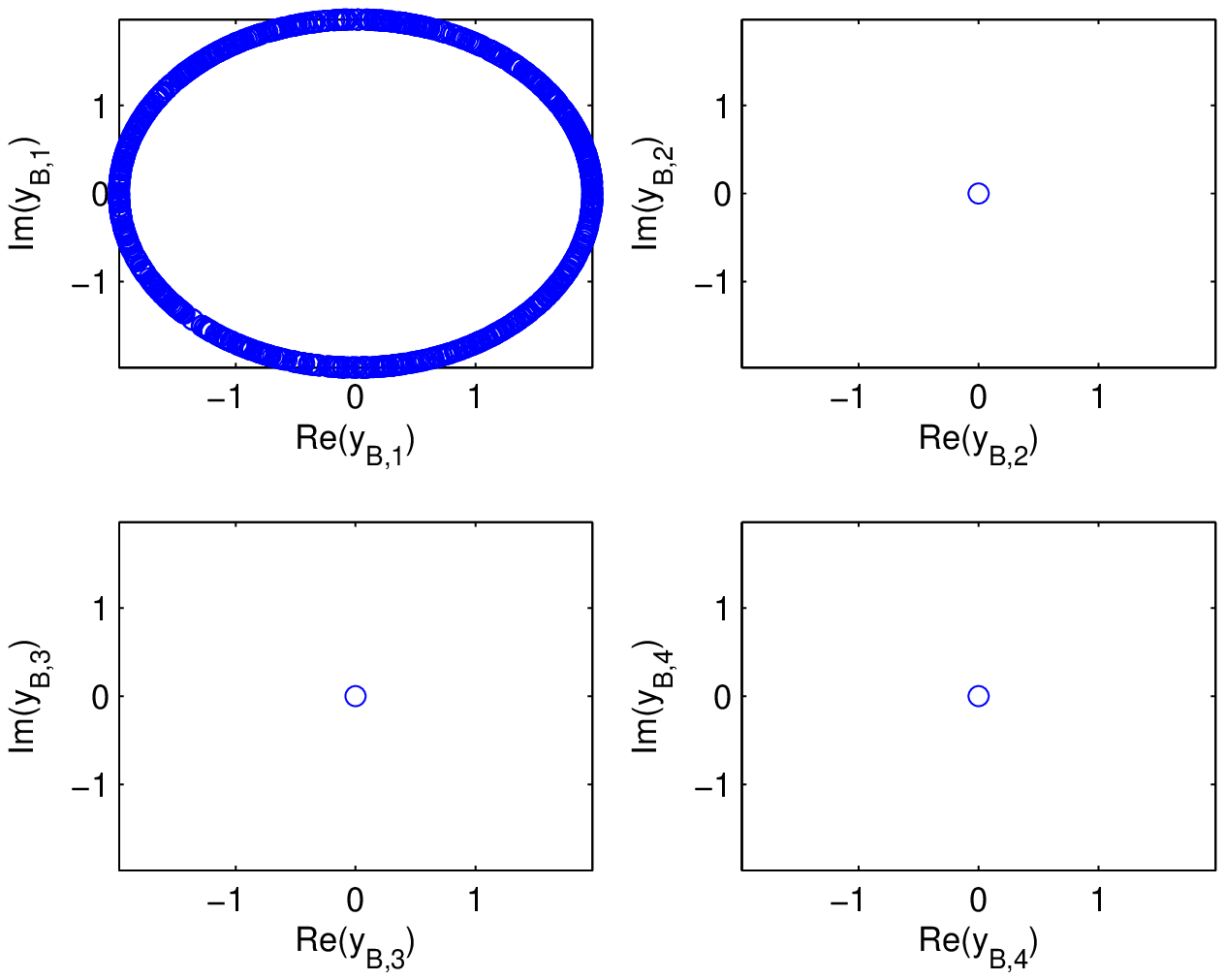}
    }
    \hfill
    \subfloat[$N_a = 1$ CAS (Eve) \label{fig_case1}]{%
      \includegraphics[width=0.485\linewidth]{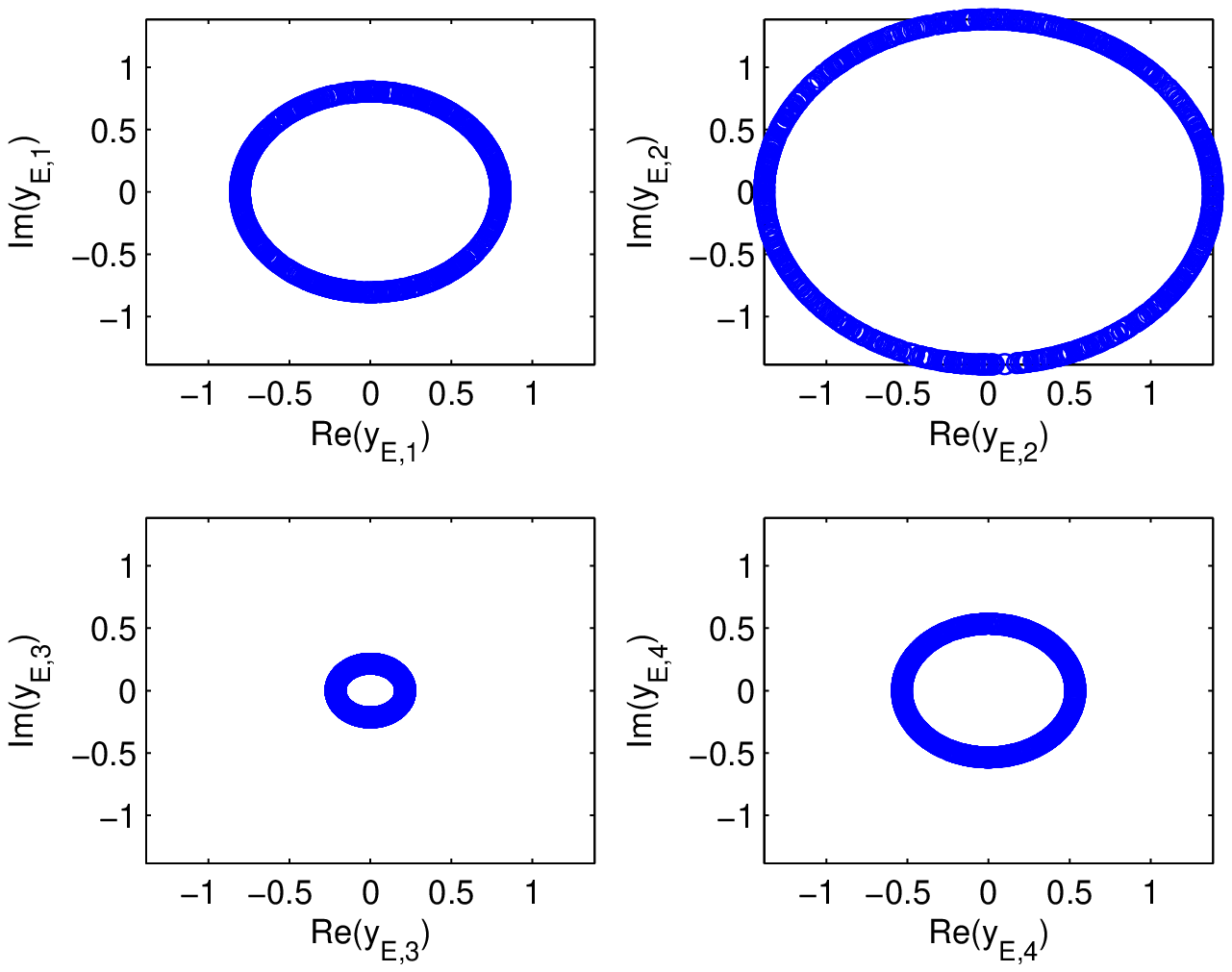}
    }
    \hfill
    \subfloat[$N_a > 1$ CAS (Bob) \label{fig_casb}]{%
      \includegraphics[width=0.485\linewidth]{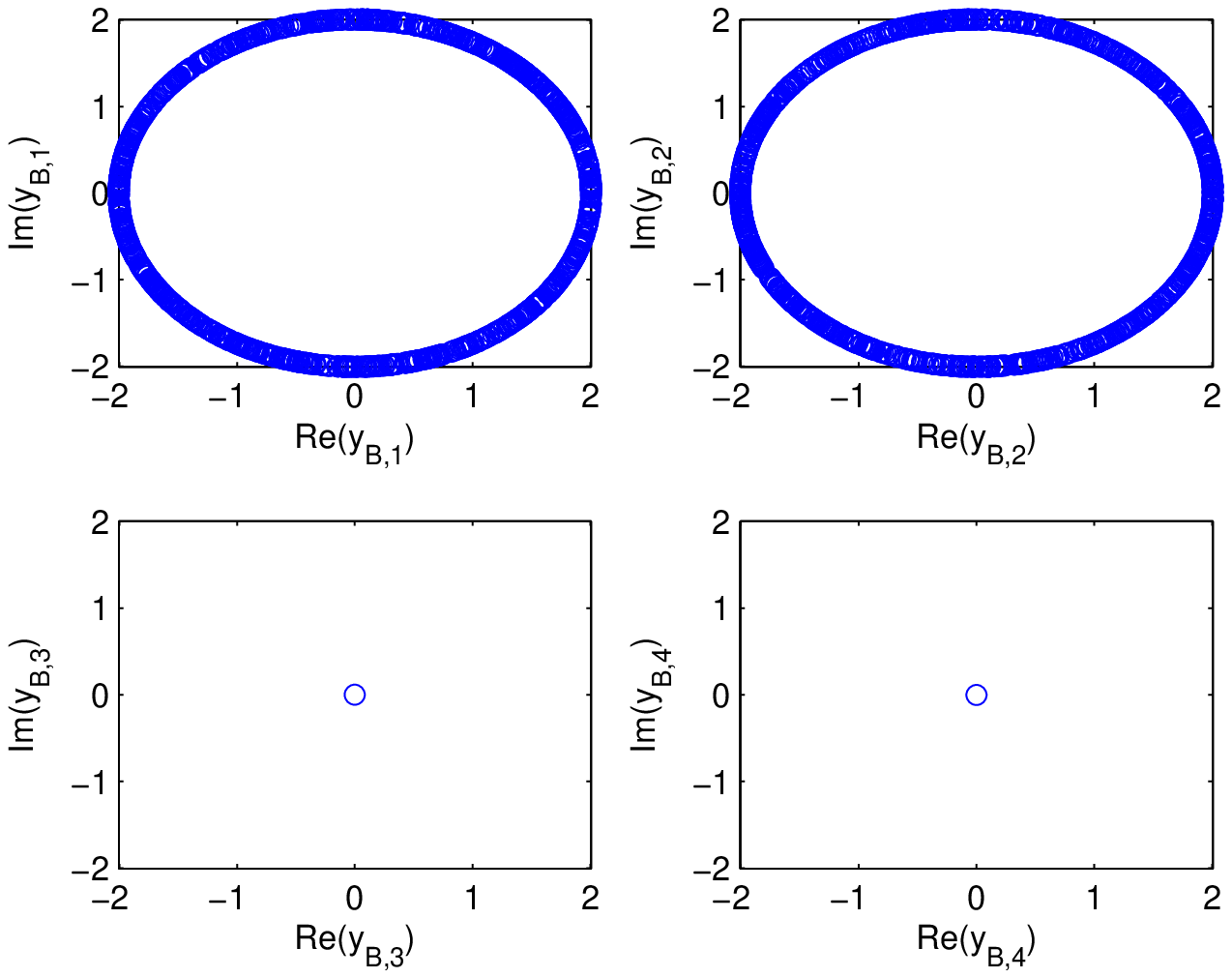}
    }
    \hfill
    \subfloat[$N_a > 1$ CAS (Eve) \label{fig_case}]{%
      \includegraphics[width=0.485\linewidth]{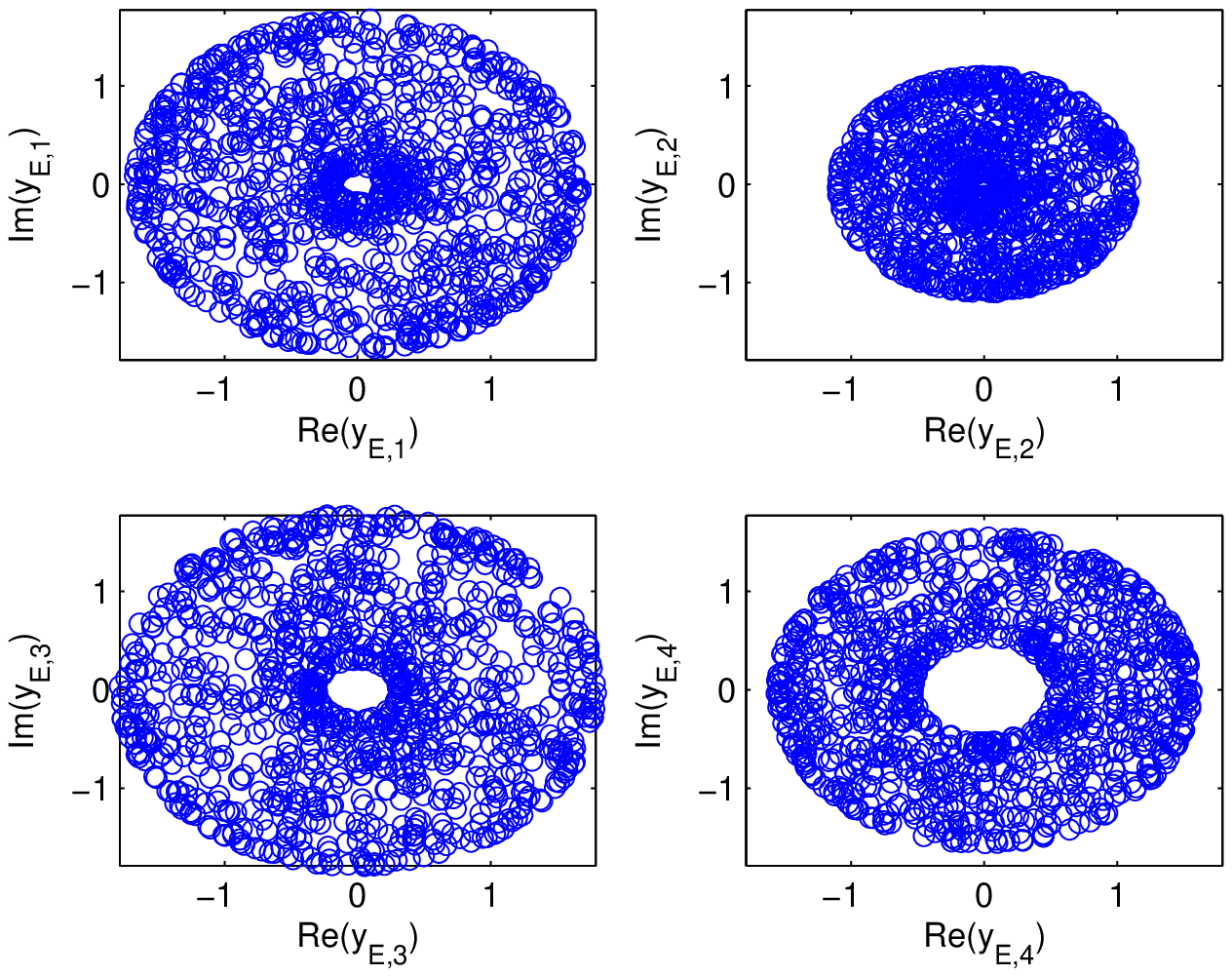}
    }
    \hfill
    \subfloat[GAS (Bob) \label{fig_gasb}]{%
      \includegraphics[width=0.485\linewidth]{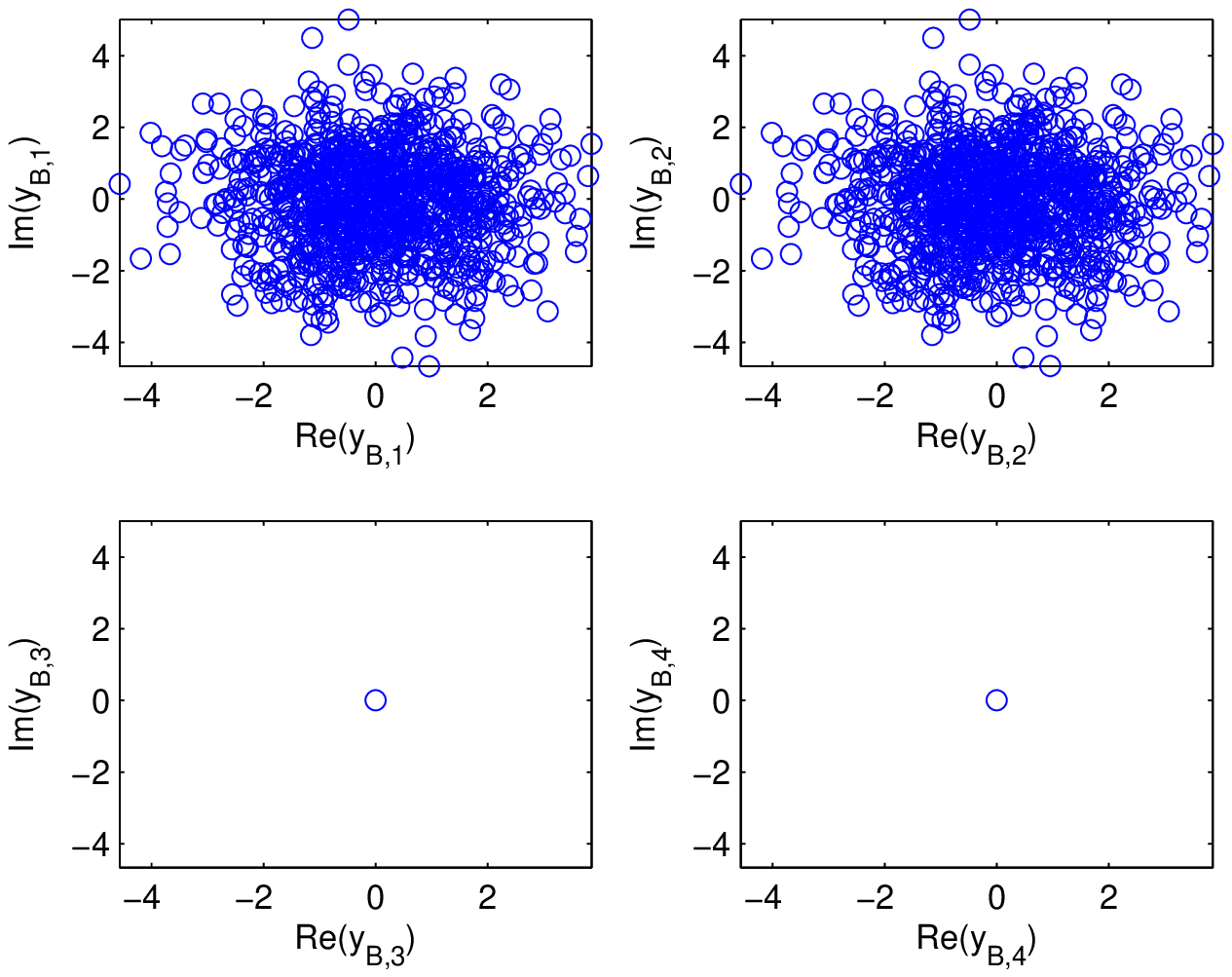}
    }
    \hfill
    \subfloat[GAS (Eve) \label{fig_gase}]{%
      \includegraphics[width=0.485\linewidth]{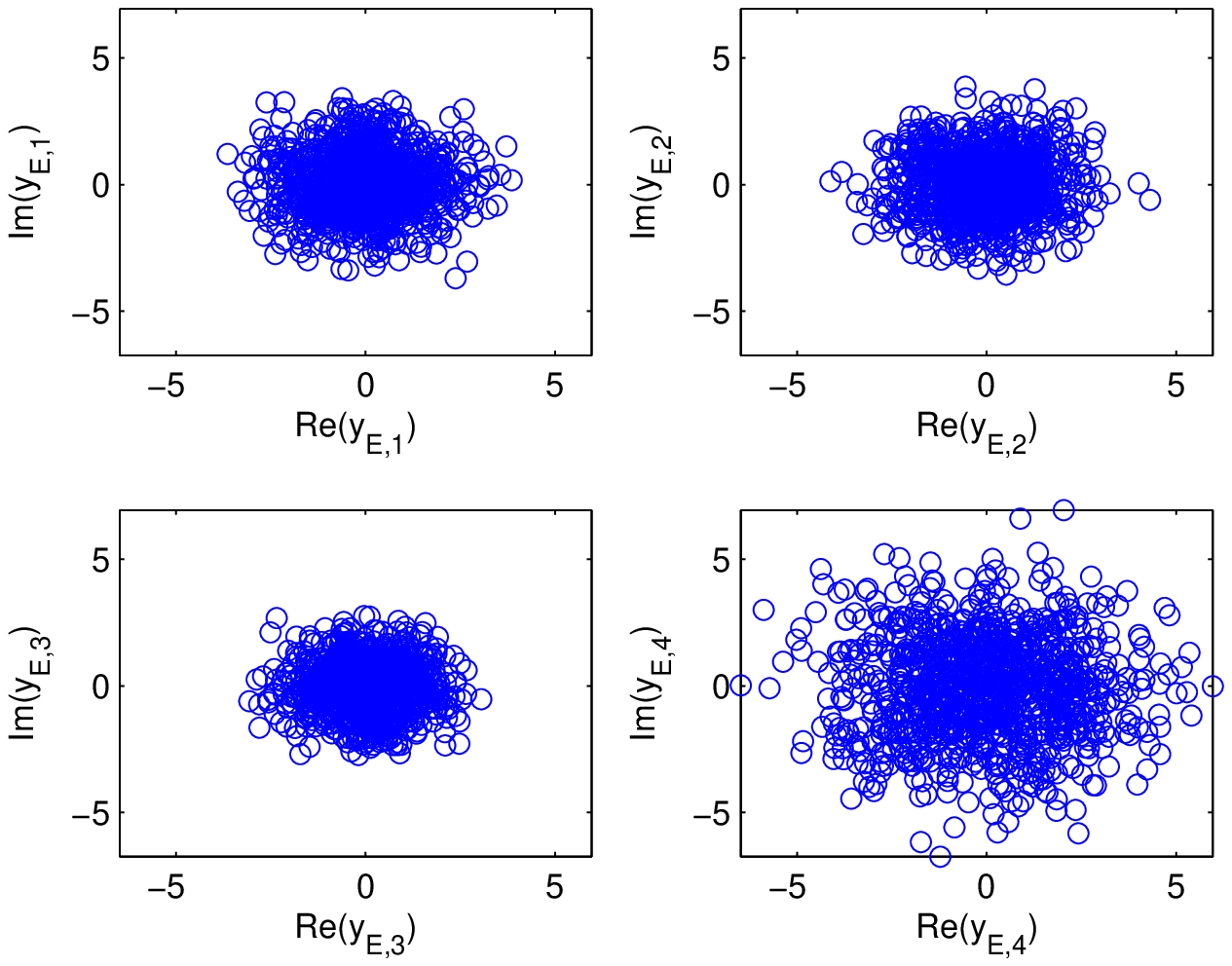}
    }
    \caption{Plot of 1000 samples of the received signal $\pmb{y}_B$ at Bob and $\pmb{y}_E$ at Eve when CI pre-coding aided GPSM is employed at Alice with $[N_t,N_r,N_e] = [8,4,4]$ under a particular MIMO channel realisation of $\pmb{H}_B$ and $\pmb{H}_E$ when $\mathrm{SNR}$ is set to 30dB. Three different cases correspond to three rows of subplots in pair for $\pmb{y}_B$ (left) and $\pmb{y}_E$ (right), respectively. On each subplot, the received signal samples of each individual RA are detailed in four miniplots, each with horizontal axis and vertical axis represent real and imaginary part.}
    \label{fig_casgas}
  \end{figure}

\subsubsection{Remarks}
We infer that the effect of CAS is analogously to introduce artificial phase ambiguity at Eve, while the effect of GAS is analogously to introduce artificial multiplicative noise. It is also worth noting that the proposed antenna scrambling is different from the so-called directional modulation~\cite{4684619,4753998,5159486}, although the useful information is also carried only the antenna level. Without using antenna scrambling as we proposed at the baseband, by varying the antenna near-field electromagnetic boundary in analogy domain, the desired phase and amplitude of the far-filed electromagnetic boundary is modulated at the intended direction, which is indistinguishable to undesired directions.

Finally, for clarity, we summarize the main findings of $\{ \Theta^B_{\epsilon,\tau}, \Theta^E_{\epsilon,\tau} \}$ for calculating the DCMC capacity $\{C_B, C_E\}$ and security capacity $C_S$ of our physical layer security for GPSM scheme with both CAS and GAS in Table~\ref{table}. We also emphasize that our proposed design imposes tough requirement on Eve by enforcing it to equip with $N_e \geq N_t$ antennas in order to start intercepting and even in this challenging setting, substantial positive security capacity can still be observed as shown in next Section.

\section{Numerical Results}
\label{sec_nr}
Let us now illustrate the numerical results of our advocated physical layer security of GPSM scheme with antenna scrambling. Throughout the simulations, we let $\sigma^2_B = \sigma^2_E > 0$ and also $N_e = N_t$ without loss of generosity. As discussed, when $N_e < N_r$, Eve is not capable of intercepting the communications between Alice and Bob, resulting into $C_E = 0$ and hence perfect security may be guaranteed by our design even when $\sigma^2_E \mapsto 0$. Thus, we focus on the challenging scenario of $N_e = N_t$ that was largely neglected in the literature. Note that this is the least number of antennas required for Eve to perform (\ref{eq_eve_post_cas}) and (\ref{eq_eve_g2}).

\subsection{Performance with Antenna Scrambling}
\subsubsection{$N_a = 1$}
The upper subplot of Fig~\ref{fig_CAS_a} shows the security capacity of GPSM scheme with both CAS (solid) and GAS (dash) employing CI based pre-coding when $N_a = 1$ having different combinations of $[N_t,N_r]$ and $N_e = N_t$. It is clear that the security capacity is higher when employing GAS for all different combinations considered, although the higher the ratio of $N_t/N_r$, the lower the difference of security capacity. Also, as expected, along with increasing the ratio of $N_t/N_r$, the security capacity increases in general for both CAS and GAS schemes.

\subsubsection{$N_a > 1$ and CAS}
The lower subplot of Fig~\ref{fig_CAS_a} shows the security capacity of GPSM with CAS employing CI based pre-coding when $[N_t, N_r] = [16,8]$ having different number of activated RAs $N_a = [2,3,4,5,6,7]$ and $N_e = N_t$. Firstly, when including $N_a = 1$ for bench-marker, the lower subplot of Fig~\ref{fig_CAS_a} shows that activating more RAs provides higher and wider security capacity region than opting for larger ratio of $N_t/N_r$ as indicated by the apparent comparisons between solid curves and dashed curves, despite that they are operated on different SNR ranges. Further, we can also see that along with increasing $N_a = [2,3,4]$, the security capacity increases. In particular, the security capacity region of activating $N_a = 2$ and $N_a = 3$ RAs at Bob constitute an approximately horizontal shifted security capacity region of activating $N_a = 6$ and $N_a = 5$ RAs. This is intuitively true since activating $N_a = 6$ ($N_a = 5$) RAs at Bob is conceptually equivalent to deactivating $N_t-N_a = 2$ ($N_t-N_a = 3$) RAs at Bob, where in both cases, we have the maximum security capacity being $C_S = k_{eff} = 4$ bits/s/Hz ($C_S = k_{eff} = 5$ bits/s/Hz). Finally, we observe that the highest security capacity is achieved when $N_a = 4$ with the maximum security capacity being $C_S = k_{eff} = 6$ bits/s/Hz. 

\begin{figure}[t]
\centering
\includegraphics[width=\linewidth]{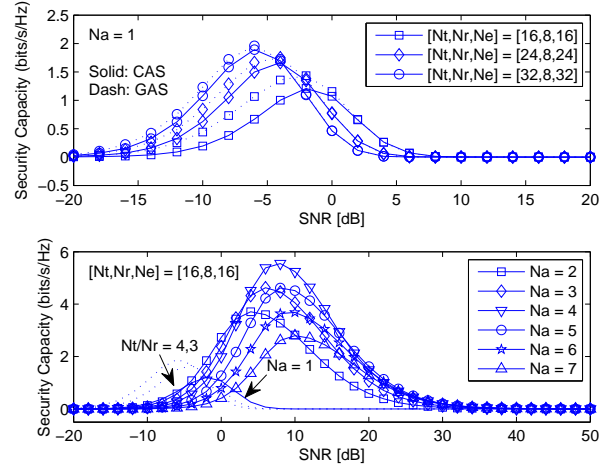}
\caption{(Upper) Security capacity of GPSM scheme with both CAS (solid) and GAS (dash) employing CI based pre-coding when $N_a = 1$ having different combinations of $[N_t,N_r]$ and $N_e = N_t$. (Lower) Security capacity of GPSM with CAS employing CI based pre-coding when $[N_t, N_r] = [16,8]$ having different number of activated RAs $N_a = [2,3,4,5,6,7]$ and $N_e = N_t$.}
\label{fig_CAS_a}
\end{figure}

\subsubsection{$N_a > 1$ and GAS}
Fig~\ref{fig_GAS} shows the security capacity of GPSM with GAS employing CI based pre-coding when $[N_t, N_r] = [16,8]$ (upper) having different number of activated RAs $N_a = [2,3,4,5,6,7]$ and $N_e = N_t$ as well as its security capacity difference with respect to that of CAS (lower). It can be seen from the upper subplot of Fig~\ref{fig_GAS} that the security capacity offered by GAS is also substantial, which is similar to the case of CAS. Again, along with increasing $N_a = [1,2,3,4]$, the security capacity increases and the mirroring effects also appear between $N_a = [1,2,3]$ and $N_a = [7,6,5]$ accordingly. It is thus interesting to investigate the security capacity difference between CAS and GAS, which is shown in the lower subplot of Fig~\ref{fig_GAS}, where we can see that the security capacity differences when $N_a > 1$ alternate around zero. This indicates no dominant superiority between CAS and GAS, where both offer substantial security capacity with different regions of operation under the setting of $N_e = N_t$.

\begin{figure}[t]
\centering
\includegraphics[width=\linewidth]{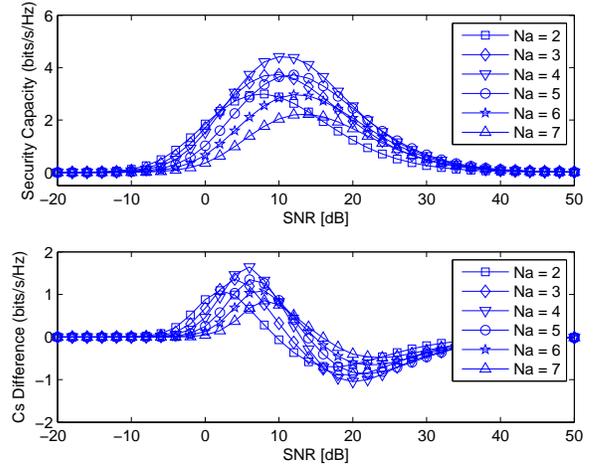}
\caption{Security capacity of GPSM with GAS employing CI based pre-coding when $[N_t, N_r] = [16,8]$ (upper) having different number of activated RAs $N_a = [2,3,4,5,6,7]$ and $N_e = N_t$ and its security capacity difference with respect to that of CAS (lower).}
\label{fig_GAS}
\end{figure}

\subsection{Further Investigations}
\subsubsection{Robustness to CSIT Error}
Like in all pre-coding schemes, an important aspect related to GPSM is its resilience to CSIT inaccuracies. In this paper, we let $\pmb{H}_B = \pmb{H}_{B,a}+\pmb{H}_{B,i}$, where $\pmb{H}_{B,a}$ represents the matrix hosting the average CSI, with each entry obeying the complex Gaussian distribution of $h_{B,a} \sim \mathcal{CN}(0,\sigma^2_{B,a})$ and $\pmb{H}_{B,i}$ is the instantaneous CSI error matrix obeying the complex Gaussian distribution of $h_{B,i} \sim \mathcal{CN}(0,\sigma^2_{B,i})$, where we have $\sigma^2_{B,a}+\sigma^2_{B,i} = 1$. As a result, only $\pmb{H}_{B,a}$ is available at Alice for pre-processing, while we assume perfect $\pmb{H}_{B}$ at Bob and Eve. The upper subplot of Fig~\ref{fig_robust} shows the security capacity of GPSM with GAS employing CI based pre-coding when $[N_t, N_r, N_a] = [16,8,2]$ and $N_e = N_t$ under CSIT error of $\sigma_i = [0.3,0.4,0.5]$. It is clear from the upper subplot of Fig~\ref{fig_robust} that, as expected, the higher the CSIT error imposed, the lower security capacity exhibited. 

\subsubsection{Robustness to Channel Correlations}
Another typical impairment is antenna correlation. The correlated MIMO channels are modelled by the widely-used Kronecker model~\cite{1021913}, which are written as $\pmb{H}_{B} =
(\pmb{R}^{1/2}_{A,t})\pmb{H}_{B,0}(\pmb{R}^{1/2}_{B,r})^T$ and $\pmb{H}_{E} =
(\pmb{R}^{1/2}_{A,t})\pmb{H}_{E,0}(\pmb{R}^{1/2}_{E,r})^T$, with $\pmb{H}_{B,0}$ and $\pmb{H}_{E,0}$ representing the original MIMO channel imposing no correlation, while $\pmb{R}_{A,t}$, $\pmb{R}_{B,r}$ and $\pmb{R}_{E,r}$ represent the correlations at the transmitter of Alice and receiver of Bob and Eve, respectively, with the correlation entries given by $R_{A,t}(i,j) = \rho_{A,t}^{|i-j|}, R_{B,r}(i,j) = \rho_{B,r}^{|i-j|}, R_{E,r}(i,j) = \rho_{E,r}^{|i-j|}$. In this paper, we assume $\rho_{A,t} = \rho_{B,r} = \rho_{E,r} = \rho$. The lower subplot of Fig~\ref{fig_robust} shows the security capacity of GPSM with GAS employing CI based pre-coding when $[N_t, N_r, N_a] = [16,8,2]$ and $N_e = N_t$ under antenna correlations of $\rho = [0.3,0.4,0.5]$. We can see from the lower subplot of Fig~\ref{fig_robust} that the security capacity of GPSM with GAS is not necessarily reduced or improved by increasing the degree of channel correlations. Rather, the positive security capacity region is shifted towards higher SNR range, whilst increasing the value of channel correlations. 

\begin{figure}[t]
\centering
\includegraphics[width=\linewidth]{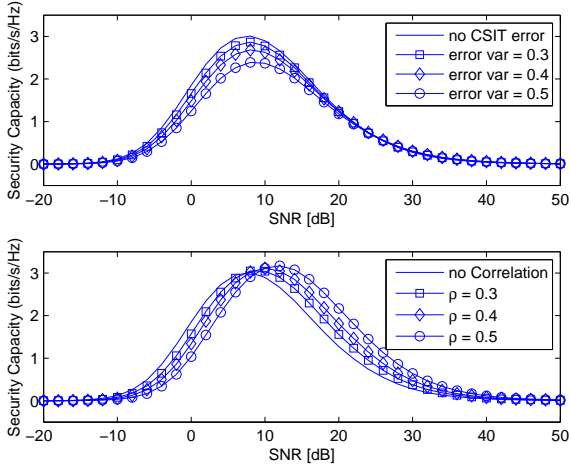}
\caption{Security capacity of GPSM with GAS employing CI based pre-coding when $[N_t, N_r, N_a] = [16,8,2]$ and $N_e = N_t$ under CSIT error of $\sigma_i = [0.3,0.4,0.5]$ (upper) and antenna correlations of $\rho = [0.3,0.4,0.5]$ (lower).}
\label{fig_robust}
\end{figure}

\subsubsection{$N_e > N_t$}
Fig~\ref{fig_large} shows the security capacity (upper) and outage security capacity (lower) when $SNR = 10$dB of GPSM with GAS employing CI based pre-coding when $[N_t, N_r, N_a] = [16,8,2]$ and $N_e = [16,18,20,22,24]$. It is clear from the upper subplot of Fig~\ref{fig_large} that the higher the number of receive antennas $N_e$ at Eve, the lower and narrower positive security capacity exhibited, as Eve's capability is naturally improved. Also, we can see from the lower subplot of Fig~\ref{fig_large} that the higher the number of receive antennas $N_e$ at Eve, the lower the span of positive security capacity values. For example, when $N_e = 22$, the security capacity ranges upto 1 bits/s/Hz, while when $N_e = 16$, the security capacity ranges upto nearly 4 bits/s/Hz. Both figures suggest that our proposed physical layer security of GPSM with antenna scrambling works reasonably well even under $N_e > N_t$.

\begin{figure}[t]
\centering
\includegraphics[width=\linewidth]{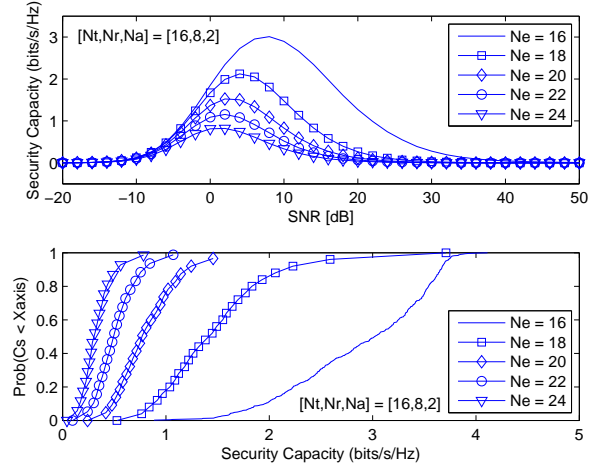}
\caption{Security capacity (upper) and outage security capacity (lower) when $SNR = 10$dB of GPSM with GAS employing CI based pre-coding when $[N_t, N_r, N_a] = [16,8,2]$ and $N_e = [16,18,20,22,24]$.}
\label{fig_large}
\end{figure}

\section{Conclusion}
\label{sec_con}

\section{List of Publications}
~\nocite{6283481,rong1,rong2,7542205,7506307,7437435,7096279,7010910,7180515,6877673,6644231,6685605,6470760,5640682,5692128,5522468,5378540,5641646,5337995,4663889,4804718,4813278,5161292,4451790,4460776,rong3,rong4,rong5,Wang2016,7802594,7542524,7465687,Wang:16,7533478,7552518,7247754,7217842,7468514,7482661,7061488,7217841,7056535,6933944,7124415,7008542,7018092,6670119,6670094,6476610,6363620,Li:13,Zhou:13,Zhou:12,6145718,6226483,6036194,4682680,5336786,Li,7249136,6399185,6214337,6213936,6093299,6093303,5779142,5779355,5779343,5594187,5594577,5594090,5594215,5502467,5493939,5208101,5073827,5073801,4698519,4657193,4533934,4525970,4489118,4224386}